\documentclass[11pt]{article}
\usepackage{graphicx,float,hyperref, indentfirst} 
\usepackage{amsmath,amsthm,amssymb,amsfonts,geometry,mathtools,xcolor, enumerate, comment}
\usepackage{varwidth}
\usepackage{subfigure}
\usepackage{tikz}
\usepackage{booktabs}
\usepackage[noend]{algpseudocode}
\usepackage{algorithm2e}
\usepackage{afterpage}
\RestyleAlgo{ruled}
\usepackage{multirow}

\SetKwFor{RepTimes}{repeat}{times}{end}
 
\theoremstyle{plain}
\newtheorem{theorem}{Theorem}
\newtheorem{lemma}[theorem]{Lemma}

\newtheorem{proposition}[theorem]{Proposition}
\newtheorem{corollary}[theorem]{Corollary}

\theoremstyle{remark}

\newcommand{\C}{\mathbb{C}}

\newcommand{\E}{\mathbb{E}}

\title{Improved Sublinear Algorithms for Classical and Quantum Graph Coloring}
\author{Asaf Ferber, Liam Hardiman, Xiaonan Chen}
\date{\today}

\begin{document}

\maketitle
\begin{abstract}
    We present three sublinear randomized algorithms for vertex-coloring of graphs with maximum degree $\Delta$. The first is a simple algorithm that extends the idea of Morris and Song to color graphs with maximum degree $\Delta$ using $\Delta+1$ colors. Combined with the greedy algorithm, it achieves an expected runtime of $O(n^{3/2}\sqrt{\log n})$ in the query model, improving on Assadi, Chen, and Khanna’s algorithm by a $\sqrt{\log n}$ factor in expectation. When we allow quantum queries to the graph, we can accelerate the first algorithm using Grover's famous algorithm, resulting in a runtime of $\Tilde{O}(n^{4/3})$ quantum queries. Finally, we introduce a quantum algorithm for $(1+\epsilon)\Delta$-coloring, achieving $O(\epsilon^{-1}n^{5/4}\log^{3/2}n)$ quantum queries, offering a polynomial improvement over the previous best bound by Morris and Song.
\end{abstract}
\section{Introduction}
Given an integer $k\in \mathbb{N}$ and a graph $G=(V,E)$, a \emph{proper $k$-coloring} is a function $c:V\rightarrow \{1,\ldots,k\}$ such that for each $\{v,u\}\in E$ we have $c(u)\neq c(v)$ (that is, no edge is \emph{monochromatic}). 
The smallest positive $k$ such that $G$ admits a proper $k$-coloring is called the \emph{chromatic number} of $G$, denoted by $\chi(G)$.
While determining $\chi(G)$ is NP-hard \cite{Karp1972}, it is easy to color a graph $G$ with maximum degree $\Delta$ using $\Delta+1$ colors. A simple greedy algorithm achieves this by assigning each vertex a color different from those of its neighbors. Since each vertex has at most $\Delta$ neighbors, there is always at least one available color to ensure no edge becomes monochromatic. Without making further assumptions about $G$, this approach cannot do with fewer than $\Delta+1$ colors: consider a clique or an odd cycle.

To better understand the performance of coloring algorithms, we can evaluate them using various models \cite{Beame_MPC,ChangCongestedClique,Goldreich_2017}. One such model is the \emph{query model}, where we assume the graph is stored in a black box and we have access to an oracle that answers the following queries.
\begin{itemize}
    \item Adjacency queries: we may ask whether the vertices $u$ and $v$ are adjacent in $V$.
    In other words, if $A$ is the adjacency matrix of $G$, we may query $A_{u,v}$ for arbitrary $u$ and $v$.

    \item Degree queries: for any vertex $v\in V$, we may query its degree, denoted $d(v)$.

    \item Neighborhood queries: for any vertex $v$ and any integer $j$, we may query the $j$-th neighbor of $v$ (the ordering of the neighbors is arbitrary, but fixed).
    If $v$ does not have $j$ neighbors, then we assume the oracle returns some special symbol $\perp$.
\end{itemize}
Here we judge an algorithm's efficiency by the number of queries it makes to such an oracle.
The aforementioned greedy algorithm makes $O(|E|)$ neighborhood queries. For dense graphs, this becomes impractical since the number of edges can be extremely large, and computational resources may not scale accordingly. Therefore, we aim to design sublinear algorithms—in terms of the number of edges—that can still efficiently color $G$.

In this direction, Assadi, Chen, and Khanna \cite{assadiChenKhannaSublinear} detailed a randomized algorithm that generates a valid $(\Delta+1)$-coloring with high probability in $O(n^{2}\log^2 n/\Delta)$ adjacency queries.\footnote{As part of the standard query model, we assume that we may query the degree of an arbitrary vertex. Consequently, determining $\Delta$ takes only $\Tilde{O}(n)$ queries, where the $\Tilde{O}$ notation omits logarithmic factors in $n$. Thus, we do not need to assume that 
$\Delta$ is known beforehand.} They use a technique known as palette sparsification, whereby we sample a small number of colors to form a palette of size $O(\log n)$ for each vertex, and then construct a valid coloring from them.
Combining their algorithm with the greedy algorithm, we achieve an algorithm that properly colors $G$ in $O(n^{3/2} \log n)$ queries.
This is just a log factor greater than the lower bound of $\Omega(n^{3/2})$ queries proved in the same paper. If one has $(1+\epsilon)\Delta$ colors, Morris and Song \cite{morrisSongColoring} recently provided a very simple randomized algorithm that outputs a valid $(1+\epsilon)\Delta$-coloring in $O(\epsilon^{-1}n^2/\Delta)$ queries in expectation. Again, when combined with the greedy algorithm, this algorithm runs in $O(\epsilon^{-1/2}n^{3/2})$ queries in expectation.

Both works use a common framework that involves two algorithms to properly color a graph $G$ with $n$ vertices and maximum degree $\Delta$. The first algorithm makes $f(\Delta)$ adjacency queries, where $f(\Delta)$ is monotonically decreasing with $\Delta$, while the second makes $g(\Delta)$ neighborhood queries, where $g(\Delta)$ is monotonically increasing with $\Delta$.
If $\Delta^*$ is chosen so that $f(\Delta^*) = g(\Delta^*)$, we can apply the first algorithm when $\Delta \geq \Delta^*$ and the second when $\Delta < \Delta^*$ to obtain a combined algorithm that runs in at most $f(\Delta^*) = g(\Delta^*)$ queries.
The improvements in above works came from enhancing the function $f$ while using $g(\Delta)=n\Delta$ from the greedy algorithm.

Following this framework, our first contribution is a simple Las Vegas algorithm for $(\Delta+1)$-coloring that improves the function $f(\Delta)$. Specifically, our algorithm uses $f(\Delta)=O(n^2\log n/\Delta)$ adjacency queries in expectation. In comparison, the algorithm of Assadi, Chen, and Khanna is a Monte Carlo algorithm that succeeds with probability \( 1 - 1/\mathrm{poly}(n) \) for some large polynomial in \( n \). However, their algorithm implicitly has an expected running time of \( f(\Delta) = O(n^2 \log^2 n / \Delta) \), so our approach improves on their result by a factor of \( \log n \) in expectation. On the other hand, our Las Vegas algorithm can be converted into a Monte Carlo algorithm using standard probabilistic boosting techniques. To achieve the same level of success probability as Assadi, Chen, and Khanna's algorithm, we require an additional \( \log n \) factor. In this case, the query complexity of our algorithm matches theirs, while still maintaining the same success probability.

Our method is relatively simple and easy to implement and analyze. Roughly speaking, we process the vertices sequentially, assigning each a random color and using adjacency queries to check if any of its neighbors share that color. For vertices with large degrees, this method can be more efficient than examining all neighbors beforehand, as randomly selecting a color still offers a good chance of avoiding a monochromatic edge. This approach is inspired by the work of Morris and Song \cite{morrisSongColoring} on $(1+\epsilon)\Delta$-coloring. By randomizing the order in which we visit the vertices of our graph, we are able to adapt it to $(\Delta+1)$-coloring. 

\begin{theorem}\label{thm: classical Delta+1}
    There is a randomized algorithm that, given a graph $G$ of maximum degree $\Delta$, properly $(\Delta+1)$-colors $G$ using $O(n^{2}\log n/\Delta)$ adjacency queries in expectation.
\end{theorem}
Using $g(\Delta)=n\Delta$ from the greedy algorithm, we derive the following corollary.
\begin{corollary}
    There is a randomized algorithm that, given a graph $G$ of maximum degree $\Delta$, properly $(\Delta+1)$-colors $G$ using $O(n^{3/2}\sqrt{\log n})$ queries in expectation.\footnote{Recently, the authors learned that this result has been independently proved by Sepehr Assadi, who presented it in a lecture at Simons Institute Bootcamp in May 2024.}
\end{corollary}

Can we do better than Theorem \ref{thm: classical Delta+1} if we have access to a quantum oracle? Following the approach of Morris and Song, we can utilize Grover's algorithm \cite{grover96}, which allows us to find a marked item in an unsorted list of $N$ items using only $O(\sqrt{N})$ quantum queries. By incorporating Grover's algorithm, we can achieve a further improvement in $f(\Delta)$, resulting in a quadratic speedup.

\begin{theorem}\label{thm: quantum adjacency}
    There is a quantum algorithm that, given a graph $G$ of maximum degree $\Delta$, properly $(\Delta+1)$-colors $G$ using $O(n^{3/2}\log n/\sqrt{\Delta})$ quantum adjacency queries in expectation .
\end{theorem}
\begin{corollary}
    There is a quantum algorithm that, given a graph $G$ of maximum degree $\Delta$, properly $(\Delta+1)$-colors $G$ using $\Tilde{O}(n^{4/3})$ quantum queries in expectation .
\end{corollary}

The $\Tilde{O}(n^{4/3})$ comes from using $g(\Delta)=n\Delta$ again from the classical greedy algorithm, similar to the approach taken by Morris and Song \cite{morrisSongColoring} for 
$(1+\epsilon)\Delta$-coloring. This result breaks the classical lower bound of $\Omega(n^{3/2})$, leading us to consider whether improvements to $g(\Delta)$ are possible using a quantum oracle. Our last contribution is a quantum algorithm that runs in $g(\Delta)=\epsilon^{-2}n\log^2 n\sqrt{\Delta}$ queries. This takes the role of the greedy algorithm in previous results. The key idea is to randomly partition the vertices and then efficiently identify the neighbors of a vertex in the same part using neighborhood queries. While achieving this efficiently in the classical setting seems unlikely, the quantum setting makes it attainable.

\begin{theorem}\label{thm: quantum neighborhood}
    There is a quantum algorithm that, given a graph $G$ of maximum degree $\Delta$, properly $(1+\epsilon)\Delta$-colors $G$ using $O(\epsilon^{-2}n\log^2 n\sqrt{\Delta})$ quantum neighborhood queries.
    This algorithm succeeds with probability at least $2/3$.
\end{theorem}
 When combined with Theorem \ref{thm: quantum adjacency}, our algorithm achieves a polynomial improvement over the results presented in \cite{morrisSongColoring}.
\begin{corollary}
    There is a quantum algorithm that, given a graph $G$ of maximum degree $\Delta$, properly $(1+\epsilon)\Delta$-colors $G$ using $\tilde{O}(\epsilon^{-1}n^{5/4})$ queries.
    This algorithm succeeds with probability at least $2/3$.
\end{corollary}

\begin{table}[h!]
    \centering
    \begin{tabular}{|l||c|c|c|c|}\hline
     \multicolumn{1}{|c||}{}& \multicolumn{2}{c|}{Classical} & \multicolumn{2}{c|}{Quantum}\\\hline
                            & Neighborhood  & Adjacency                 & Neighborhood              & Adjacency\\\hline\hline
     $(\Delta+1)$-coloring          & $n\Delta$     & $O\left(\frac{n^2 \log^2 n}{\Delta}\right)$\cite{assadiChenKhannaSublinear}, $O\left(\frac{n^2\log n}{\Delta}\right)$[$\star$]       &                   & $\tilde O \left(\frac{n^{3/2}}{\sqrt \Delta}\right)$[$\star$]\\\hline
     $(1+\epsilon)\Delta$-coloring  &               & $\frac{n^2}{\epsilon \Delta}$\ \cite{morrisSongColoring} & $\tilde O\left(\frac{n\sqrt \Delta}{\epsilon^2}\right)$[$\star$] & $\tilde O\left( \frac{n^{3/2}}{\epsilon^{3/2}\sqrt \Delta} \right)$\cite{morrisSongColoring}\\\hline
    \end{tabular}
    \caption{Summary of expected runtimes. A [$\star$] indicates this work.}
    \label{tab:my_label}
\end{table}

We provide a table summarizing the current best algorithms using neighborhood or adjacency queries. Blank cells indicate areas where no non-trivial algorithms have been discovered so far. Discovering new algorithms that fill these gaps or improving existing algorithms using only one type of query could lead to improvements in the combined algorithm. In the remainder of the paper, we prove Theorem \ref{thm: classical Delta+1} in Section 2 and Theorems \ref{thm: quantum adjacency} and \ref{thm: quantum neighborhood} in Section 3.

\section{Classical Sublinear Algorithm}
Let $V = \{v_1, \ldots, v_n\}$ be the vertices of $G$. We provide an algorithm that constructs the coloring $\chi: V\to [\Delta+1]$ by successively updating the color classes $\chi^{-1}(1), \ldots, \chi^{-1}(\Delta+1)$.
Start by selecting a permutation $\sigma \in S_n$ uniformly at random, and set $\chi_0^{-1}(c) = \emptyset$ for all colors $c\in[\Delta+1]$. At each step $0 < t \leq n$, choose a color $c$ 
uniformly at random for $v_{\sigma(t)}$. Let $\chi_t^{-1}(c)$ denote the set of vertices that have already been assigned the color $c$ up to step $t$. If the vertex $v_{\sigma(t)}$ is adjacent to any vertex in $\chi_t^{-1}(c)$, select a new random color $c'$. Otherwise, assign the color $c$ to $v_{\sigma(t)}$, and update $\chi_t^{-1}(c)$ by adding $v_{\sigma(t)}$ to the color class.
\begin{algorithm}[h]\label{alg: max degree color}
\caption{$(\Delta+1)$-Color($G,\Delta$)}
Choose a permutation $\sigma\in S_n$ uniformly at random

\For{$t = 1, \ldots, n$}{
    \While{$v_{\sigma(t)}$ is not colored}{
        Choose a color $c$ uniformly at random in $[\Delta + 1]$
        
        \For{$u \in \chi_t^{-1}(c)$}{
            Query if $\{u,v_{\sigma(t)}\}\in E$ 
            
            \If{$\{u,v_{\sigma(t)}\}\in E$}{
                \textbf{break}
            }
            Assign color $c$ to $v_{\sigma(t)}$ and update $\chi(c)$ with $v_{\sigma(t)}$
        }
    }
}
\Return the color assignment $\chi: V\to [\Delta+1]$
\end{algorithm}

\begin{proof}[Proof of Theorem \ref{thm: classical Delta+1}]
    
Let us analyze the expected runtime of Algorithm \ref{alg: max degree color}.
For each $t\leq n$, let $Q_t$ be the number of times we query the adjacency matrix before we assign a valid color to $v_{\sigma(t)}$, and let $Q_t(\sigma)$ be this random variable conditioned on the event that we ordered $V$ according to the permutation $\sigma$.
If we let $c$ be the first color we try to assign to $v_{\sigma(t)}$, 
let $p_t(\sigma)$ be the probability that it is a valid color, i.e., that no neighbor of $v_{\sigma(t)}$ has already been assigned the color $c$.
In the worst case, we need to check $v_{\sigma(t)}$ for adjacency with every vertex in $\chi^{-1}_t(c)$, so we can bound the expectation of $Q_t(\sigma)$ as follows.
\[
\E_c[Q_t(\sigma)] \leq p_t(\sigma)\E_c[|\chi^{-1}_t(c)| : c\text{ is valid}] + (1-p_t(\sigma))\bigg( \E_c[|\chi^{-1}_t(c)|: c\text{ is invalid}] + \E_c[Q_t(\sigma)]\bigg),
\]
so $\E_c[Q_t(\sigma)] \leq \E_c[|\chi^{-1}_t(c)|]/p_t(\sigma)$.
Since each vertex is colored randomly, we have
\[
\E_c [Q_t(\sigma)] \leq \frac{\E_c[|\chi^{-1}(c)|]}{p_t(\sigma)} \leq \frac{n}{\Delta+1}\cdot \frac{1}{p_t(\sigma)}.
\]
We obtain the expectation of $Q_t$ by averaging $Q_t(\sigma)$ over a uniformly chosen $\sigma$:
\[
\E[Q_t] \leq \frac{n}{\Delta+1}\cdot\E_\sigma \left[\frac{1}{p_t(\sigma)}\right].
\]
Now when we attempt to color $v_{\sigma(t)}$ by choosing $c$ at random, the worst case is that each neighbor of $v_{\sigma(t)}$ that has already been colored has a different color.
In other words, if we let
$$L_t(\sigma) = |\{i < t: \{v_{\sigma(i)},v_{\sigma(t)}\} \in E\}|$$
denote the number of neighbors of $v_{\sigma(t)}$ that have already been colored by time $t$, then
\[
p_t(\sigma) \geq \frac{\Delta + 1 - L_t(\sigma)}{\Delta+1}.
\]
The expectation is then bounded by
\begin{equation}\label{eqn:Qt exp bound}
\E[Q_t] \leq n\cdot  \E_\sigma\left[ \frac{1}{\Delta+1 - L_t(\sigma)}\right] = n\sum_{k=0}^\Delta \frac{\Pr[L_t(\sigma) = k]}{\Delta+1 - k}.
\end{equation}
After counting the number of permutations where $v_{\sigma(t)} = v$ and exactly $k$ neighbors of $v$ come before $v$ with respect to $\sigma$, the probability in the summation is given by
\[
\Pr[L_t(\sigma) = k] = \frac{1}{n}\sum_{v\in V}\Pr[L_t(\sigma) = k \mid v_{\sigma(t)} = v] = \frac{1}{n}\sum_{v\in V} \frac{\binom{t-1}{k}\binom{n-t}{d(v)-k}}{\binom{n-1}{d(v)}}.
\]
Substituting this into (\ref{eqn:Qt exp bound}) and switching the order of summation gives
\begin{equation}\label{eqn:expectation sum}
\E[Q_t] \leq \sum_{v\in V}\frac{1}{\binom{n-1}{d(v)}}\sum_{k = 0}^{d_v}\binom{t-1}{k}\binom{n-t}{d(v)-k} \frac{1}{\Delta+1-k}.
\end{equation}

Now if $d(v)$ is small, say less than $\epsilon \Delta$ for some small positive $\epsilon$, then
\begin{equation}\label{eqn: large degree}
\sum_{k=0}^{d(v)}\binom{t-1}{k}\binom{n-t}{d(v)-k}\frac{1}{\Delta+1-k} \leq \frac{1}{(1-\epsilon)\Delta}\sum_{k=0}^{d(v)}\binom{t-1}{k}\binom{n-t}{d(v)-k} = \frac{1}{(1-\epsilon)\Delta}\binom{n-1}{d(v)}.
\end{equation}
The last equality comes from recognizing the middle sum as the coefficient of $x^{d(v)}$ in the expansion of $(1+x)^{t-1}\cdot (1+x)^{n-t} = (1+x)^{n-1}$.
On the other hand, if $d(v)\geq \epsilon \Delta$, then
\begin{equation}\label{eqn: small degree}
\begin{split}
    \sum_{k=0}^{d(v)}\binom{t-1}{k}\binom{n-t}{d(v)-k}\frac{1}{\Delta+1-k} & \leq\sum_{k=0}^{d(v)}\binom{t-1}{k}\binom{n-t}{d(v)-k}\frac{1}{d(v)+1-k}.
\end{split}
\end{equation}
In order to bound this term, we notice that
\[
\sum_{k=0}^{\infty}\binom{n-t}{k}\frac{x^{k}}{k+1}=\frac{1}{x}\int_0^x \sum_{k=0}^{\infty}\binom{n-t}{k}s^kds=\frac{1}{x}\int_0^x (1+s)^{n-t}ds=\frac{(1+x)^{n-t+1}}{x(n-t+1)}.
\]
Hence the right-hand side in (\ref{eqn: small degree}) may be recognized as the coefficient of $x^{d(v)}$ in the expansion of 
\[
(1+x)^{t-1}\cdot\frac{(1+x)^{n-t+1}}{x(n-t+1)}=\frac{(1+x)^n}{x(n-t+1)},
\]
which is $\frac{1}{n-t+1}\binom{n}{d(v)+1}$.
Therefore,
\begin{align*}\label{eqn: Qt bound}
       \E[Q_t]&\le\sum_{v:\ d(v)<\epsilon \Delta}\frac{\binom{n-1}{d(v)}}{\binom{n-1}{d(v)}}\frac{1}{(1-\epsilon)\Delta}+\sum_{v:\ d(v)\ge \epsilon\Delta}\frac{\binom{n}{d(v)+1}}{\binom{n-1}{d(v)}}\frac{1}{n-t+1}\\
       &\le\sum_{v:\ d(v)<\epsilon \Delta}\frac{1}{(1-\epsilon)\Delta}+\frac{n}{n-t+1}\sum_{v:\ d(v)\ge \epsilon\Delta}\frac{1}{d(v)+1}.
\end{align*}
Notice $\epsilon$ here can be any constant, so the total number of query we need is 
\begin{align*}
    \E[\# \textrm{ of queries}]=\sum_{t=1}^n \E[Q_t]\le O(1)\cdot \frac{n}{\Delta}\cdot \sum_{t=1}^n\frac{n}{n-t+1} =O\left(\frac{n^2\log n}{\Delta}\right).
\end{align*}
\end{proof}
It is worth noting that this Las Vegas algorithm can be readily converted to a Monte Carlo algorithm that succeeds with probability 
$1-1/\mathrm{poly}(n)$, at the cost of an additional $\log n$ factor. By Markov's inequality, 
\[\Pr(\textrm{Algorithm 1 fails after }2\E[\# \textrm{ of queries}])\le 1/2.\]
Repeat this algorithm if it fails for $k\log_2 n$ times, the probability of failing all the time is bounded by $1/n^k$ for any constant $k$. 



\section{Quantum Sublinear Algorithms}

In quantum computing, the fundamental unit of information is the qubit. Unlike a classical bit, which exists strictly in one of two states, 0 or 1, a qubit can exist in a superposition of both states simultaneously. Formally, a qubit is represented as a unit vector in the complex vector space $\C^2$.
Using Dirac’s ``bra-ket'’ notation, we denote a general qubit as $|\psi\rangle$ and denote the basis vectors for a pre-fixed orthonormal basis of $\C^2$ by $|0\rangle$ and $|1\rangle$.
Observing a qubit collapses its superposition to a definite basis state, $|0\rangle$ or $|1\rangle$. If we let $|\psi\rangle=\alpha|0\rangle+\beta|1\rangle$, then measuring a qubit yields $|0\rangle$ with probability $|\alpha|^2$ and $|1\rangle$ with probability $|\beta|^2$. For $n$ qubits, the system exists as a superposition of $2^n$ classical states, i.e \[|\psi\rangle = \sum_{0\le x\le 2^n-1}\alpha_x|x\rangle_n\textrm{, where }\sum_{0\le x\le 2^n-1}|\alpha_x|^2=1.\]
Qubits are manipulated by quantum gates, which can be represented as unitary transformations. Any operation of $n$ qubits can thus be described by a unitary transformation acting on $\C^{2^n}$. We refer the reader to \cite{Mermin_2007} for more details on quantum computing. 

In our problem, we transition from classical to quantum queries. Formally, any query problem can be viewed as evaluating an unknown function \( f: \{0,1\}^n \to \{0,1\}^m \), where \( f \) acts as an oracle, or black box, and each evaluation is called a query. The complexity is determined by the number of times \( f \) is queried. In the quantum setting, a standard approach to implementing \( f \) is to replace it with a unitary operator \( \mathcal{O}_f \) that acts on two registers: one for domain points and one for image points:
\[
\mathcal{O}_f |x, y\rangle = |x, y \oplus f(x)\rangle.
\]
Here, \( \oplus \) denotes the bitwise XOR operation. If \( f \) is binary-valued, we can also implement it as
\[
\mathcal{O}_f |x\rangle = (-1)^{f(x)} |x\rangle.
\]
To obtain this, we can apply the first oracle to $|x,-\rangle$, where $|-\rangle=\frac{1}{\sqrt{2}}(|0\rangle-|1\rangle)$, and only measure the first register. This conversion is sufficient for all known examples.

For instance, we consider a quantum adjacency query defined by the unitary map \( \mathcal{O}_A \) on the space \( \mathbb{C}^V \otimes \mathbb{C}^V \), given by
\[
\mathcal{O}_A |v, u\rangle = (-1)^{f_A(v, u)} |v, u\rangle,
\]
where \( f_A(v, u) = 1 \) if \( v \) and \( u \) are adjacent, and \( f_A(v, u) = 0 \) otherwise. Thus, when we refer to the number of quantum adjacency queries to \( G \), we mean the number of applications of \( \mathcal{O}_A \).

Morris and Song \cite{morrisSongColoring} observed that we can use Grover's \cite{grover98} algorithm to speed up the step of Algorithm \ref{alg: max degree color} where we check to see if vertex $v_{\sigma(t)}$ has a neighbor in the color class $\chi_t^{-1}(c)$ (they consider an algorithm for $(1+\epsilon)\Delta$-coloring that uses this same check as a subroutine).
Originally designed to find a unique marked item in an unsorted list of $N$ items,  Boyer, Brassard, H{\o}yer and Tapp \cite{boyer1998tight} showed that repeated applications of Grover's original algorithm could be used to find one of possibly multiple marked items, even when the number of marked items is unknown.

\begin{theorem}[Grover's Algorithm \cite{boyer1998tight}]\label{thm: Grover}
    Suppose we have an unsorted list of $N$ items, $k \geq 0$ of which are marked.
    Then there is a quantum algorithm that does the following:
    If $k > 0$, the algorithm outputs one marked item uniformly at random, using $O(\sqrt{N/k})$ quantum queries.
    If $k = 0$, then the algorithm concludes that there is no marked item using $O(\sqrt N)$ queries.
\end{theorem}
This will be the main tool we use to build our quantum algorithms. We will prove theorem \ref{thm: quantum adjacency} in section 3.1 and prove theorem \ref{thm: quantum neighborhood} in section 3.2

\subsection{Quantum algorithm for $(\Delta+1)$-coloring}
Recall that in our algorithm \ref{alg: max degree color}, we need to determine whether a vertex 
$v_{\sigma(t)}$ is adjacent to any vertex in $\chi_t^{-1}(c)$ — that is, we check if any vertex assigned color $c$ in step $t$ is a neighbor of $v_{\sigma(t)}$. In the quantum version of the algorithm, the only modification occurs in this checking step: we replace the classical search process with Grover's algorithm. 
\begin{proof}[Proof of Theorem \ref{thm: quantum adjacency}]
We apply the same argument as in the proof of Theorem \ref{thm: classical Delta+1}. In the classical setting, to color vertex $v_{\sigma(t)}$ with a randomly chosen color $c$, we bound the expected number of adjacency queries used to find a  neighbor of $v_{\sigma(t)}$ in $\chi_t(c)$ by $n/(\Delta+1)$.  Then
\[
\E_c [Q_t(\sigma)] \leq \frac{n}{\Delta+1}\cdot \frac{1}{p_t(\sigma)}.
\]
We aim to improve this bound in the quantum setting.
Specifically, we can view this checking process as an instance of the unstructured search problem in a list of size $|\chi_t^{-1}(c)|$. Applying Theorem \ref{thm: Grover}, we obtain
\begin{align*}
\E_c[Q_t(\sigma)]& \leq \frac{1}{p_t(\sigma)}\cdot  \E_c\left[O\left(\sqrt{|\chi_t^{-1}(c)|}\right)\right]\\
&\leq \frac{1}{p_t(\sigma)} \cdot O\left(\sqrt{\E_c[|\chi_t^{-1}(c)|]}\right)\\
&= O\left(\sqrt{\frac{n}{\Delta+1}}\cdot \frac{1}{p_t(\sigma)}\right).
\end{align*}
Here we used Jensen's inequality to move to the second line.
Taking the expectation over $\sigma$ and following the proof of Theorem \ref{thm: classical Delta+1} gives $O(\frac{n^{3/2}\log n}{\sqrt \Delta})$ quantum adjacency queries.

\end{proof}

We summarize this algorithm in Algorithm \ref{alg: quantum adjacency color}.
\vspace{3mm}

\begin{algorithm}[H]\label{alg: quantum adjacency color}
\caption{QuantumAdjacencyColor($G,\Delta$)}
Choose a permutation $\sigma\in S_n$ uniformly at random

\For{$t = 1, \ldots, n$}{
    \While{$v_{\sigma(t)}$ is not colored}{
        Choose a color $c$ uniformly at random in $[\Delta + 1]$

        \If{Grover's algorithm finds a neighbor of $v_{\sigma(t)}$ in $\chi_t^{-1}(c)$}{
            \textbf{break} choose a new color
        }
        Assign color $c$ to $v_{\sigma(t)}$ and update $\chi(c)$ with $v_{\sigma(t)}$ 
            
    }
}
\Return the color assignment $\chi: V\to [\Delta+1]$
\end{algorithm}

\subsection{Quantum algorithm for $(1+\epsilon)\Delta$-coloring}
Can we achieve a similar quantum speedup using neighborhood queries? Our motivating example is the following: Let \( G = G_{n,p} \) be the Erd\H{o}s-R\'enyi random graph, where each pair of vertices (independently) forms an edge with probability \( p \). Partition the vertex set \( V(G) = [n] \) into \( t \) subsets \( V_1, \ldots, V_t \) (assuming without loss of generality that \( t \mid n \)), where for each \( i \),
\[
V_i = \left[\frac{(i-1)n}{t} + 1, \frac{in}{t} \right].
\]
 
Let's suppose that the neighborhood queries respect this vertex ordering—i.e., querying the \(i\)-th and \(j\)-th neighbors of \(v\) returns vertices \(u\) and \(w\), respectively, with \(u \le w\) when \(i \le j\).
Each vertex $v$ in $V_i$ has at most $(1+\epsilon)np/t$ neighbors in $V_i$ with high probability by the Chernoff bound, so the set of indices $I = \{1 \leq i \leq n: \text{the }i\text{-th neighbor of }v\text{ is in }V_i \}$ forms an interval of length at most $(1+\epsilon)np/t$.
We can then efficiently find the ``first'' neighbor of \(v\) within $G[V_i]$ using binary search and then query the next $(1+\epsilon)np/t$ neighbors to learn all of $N(v)\cap V_i$.

By repeating this process for each $v\in V_i$, we learn the entire edge set of $G[V_i]$.
We can now color each subgraph using a distinct color palette of size $(1+\epsilon)np/t$, thereby obtaining a proper $(1+\epsilon)np$-coloring of $G$ using only $O(n^2p/t)$ queries. For concentration of the degrees, we need $t=np/\log n$ and results an $O(n \log n)$ algorithm.  

This approach works well for random graphs because the neighbors of any particular vertex are more or less evenly distributed throughout the graph.
In the standard setting where neighborhood queries do not adhere to some fixed order, however, querying a specific interval of a vertex \(v\)’s neighbors may not capture neighbors that belong to the same part as \(v\).
If we could somehow pluck the neighbors of $v$ that live in $V_i$ from the neighborhood list of $v$, then we might be able to construct a proper coloring in fewer queries than the classical greedy algorithm.
We will show that Grover's algorithm allows us to accomplish exactly this in the quantum setting.


To begin, we define the oracle used in our quantum algorithm.
For any vertex $v\in V(G)$, $1\le j\le \Delta$, let the query function $f_N:V\times [\Delta]\to V\cup \{0\}$ be given by 
\[f(v,j)= \begin{cases}
u_j, &\text{ if $u_j$ is the $j$-th neighbor of $v$ }\\
0, &\text{ otherwise. }
\end{cases}\]
Then a quantum neighborhood query can be implemented by the map $\mathcal O_N$, which operates on $V\otimes [n]\otimes (V\cup \{0\})$ by
\begin{equation}\label{eqn: quantum neighborhood query}
\mathcal O_N |v, j, u\rangle = |v,j,u\oplus f(v,j)\rangle.
\end{equation}
This map is unitary, and moreover, self-inverse, i.e., $\mathcal O_N\circ \mathcal O_N=I_{V\otimes [n]\otimes (V\cup \{0\}) }$. 

Now we proceed with our algorithm.
We start by demonstrating that, when we randomly and equitably partition $V(G)$ into $t$ parts, that is, each part has size either $\lceil n/t\rceil$ or $\lfloor n/t\rfloor$, then a roughly $\frac 1t$-fraction of a vertex' neighborhood follows it to its part.


\begin{lemma}\label{lem: chernoff_degree}
    Fix $\epsilon > 0$ and let $G = (V, E)$ be a graph on $n$ vertices with maximum degree $\Delta$.
    Set $t = \frac{\epsilon^2\Delta}{6\log n}$.
    If we equitably partition $V = V_1 \cup \cdots \cup V_t$ at random, then with probability at least $1-\frac 1n$, the maximum degree of each $G[V_i]$ is at most $(1+\epsilon)\frac \Delta t$.
\end{lemma}
\begin{proof}
    Let $d_i(v)$ denote the number of neighbors of $v$ in $G[V_i]$ for each $v\in V_i$. Then $\E[d_i(v)]\le \frac{\Delta}{t}$.
    By the Chernoff bound, with $t=\frac{\epsilon^2\Delta}{6\log n}$
    \begin{align*}
        \Pr\left(d_i(v)\ge (1+\epsilon)\frac{\Delta}{t}\right)\le \exp\left(-\frac{\epsilon^2\Delta}{3t}\right)\le \frac{1}{n^2}.
    \end{align*}
    Now union bound over all $n$ vertices. 
\end{proof}

After partitioning $V(G)$ as in the above lemma, our plan is to assign each part its own color palette and mimic the greedy coloring algorithm in each part.
With this in mind, consider some $v\in V_i$.
With probability at least $1/4$, it takes $O(\sqrt \Delta)$ quantum neighborhood queries to Grover search $N(v)$ for a neighbor $u$ of $v$ that also lives in $V_i$ (or conclude that there is no such neighbor).
Moreover, the neighbor that Grover's algorithm outputs is random and uniformly distributed over all neighbors in $V_i$.
Finding all neighbors of $v$ that live in $V_i$ then amounts to running a sort of coupon collector process, with the twist that we have a $1/4$ chance of receiving a coupon at all.
We are still able to obtain a useful tail estimate despite this twist.

\begin{lemma}\label{lem: coupon collector}
    Consider a process where at each step, an event occurs with success probability $p$ and upon a success, one of $k$ distinct outcomes is chosen uniformly at random. Let $T$ be the number of steps required until all $k$ distinct outcomes have been observed at least once. Then for any positive constant $C$,
    \[
    \Pr\left[T \geq \frac Cp k \log k\right] \leq k^{-C+1}.
    \]
\end{lemma}
\begin{proof}
    The probability that the $i$-th outcome has not been observed after step $t$ is
    \[
    \left(1 - \frac pk  \right)^t \leq e^{-pt/k}.
    \]
    Union bounding over all $k$ coupons and setting $t = \frac Cp k \log k$ proves the result.
\end{proof}

As $v$ has at most $(1+\epsilon)\frac \Delta t$ neighbors in $V_i$ (with high probability), Lemma \ref{lem: coupon collector} shows that we need at most $O_\epsilon(\frac{\Delta}{t} \log \frac{\Delta}{t})$ Grover searches to find all of them.
Once we have these neighbors, we can assign $v$ a color not seen among them from the palette assigned to $V_i$.
We collect the details of this procedure in Algorithm \ref{alg: quantum neighborhood color}, \textbf{QuantumNeighborhoodColor}, and its subroutine, Algorithm \ref{alg: grover neighbors}, \textbf{GroverNeighbors}.
The proof of its correctness and analysis of its runtime will comprise the proof of Theorem \ref{thm: quantum neighborhood}.
\vspace{3mm}

\begin{algorithm}[H]\label{alg: quantum neighborhood color}
\caption{QuantumNeighborhoodColor($G, \epsilon$)}
Let $\mathcal P = \{V_1, \ldots, V_t\}$ be a random equitable partition of $V(G)$ into $t = \frac{\epsilon^2\Delta}{6\log n}$ parts

Let $P_1, \ldots, P_t$ be pairwise disjoint color palettes, each with $(1+\epsilon)\frac{\Delta}{t}$ colors

Prepare the operator $W = W(\mathcal P)$ as described in (\ref{eqn: partition operator})

\For{$i = 1, \ldots, t$}{
    \For{$v\in V_i$}{
        $N_i(v) \gets$GroverNeighbors$(G, v, \epsilon, \mathcal P, W)$

        Set $\chi(v)$ to an arbitrary color from $P_i$ not seen among those in $\chi^{-1}(N_i(v))$
    }
}

\Return The color assignment $\chi: V\to (P_1 \cup \cdots \cup P_t)$
\end{algorithm}

\SetKwFor{RepTimes}{repeat}{times}{end}

\begin{algorithm}[H]\label{alg: grover neighbors}
\caption{GroverNeighbors($G, v, \epsilon, \mathcal P, W$)}

$d\gets d(v)$

$k \gets \min\big( (1+\epsilon)\frac{\Delta}{t}, d  \big)$

$N_i(v) \gets \emptyset$

$|U\rangle \gets d^{-1/2}\sum_{j=1}^{d}|v, j, 0\rangle$

Prepare the operators $R_U = 2|U\rangle \langle U| - I$ and $R_B = \mathcal O_N W\mathcal O_N$

\RepTimes{$8k \log k\cdot \frac{\log n}{\log \log n}$}{
    $m\gets 1$
    
    \While{$m \leq \sqrt{d}$}{
        Set $j$ to a uniformly chosen random integer between 0 and $m$
    
        $|\psi\rangle \gets d^{-1/2}\sum_{j=1}^{d}|v, j, 0\rangle$
    
        $|\psi\rangle \gets (R_UR_B)^j|\psi\rangle$
    
        Measure $|\psi\rangle$ in the $\{|v, j, 0\rangle: 1\leq j\leq d\}$ basis to obtain $|v, j^*, 0\rangle$
    
        Let $u$ be the result of the classical neighborhood query for the $j^*$-th neighbor of $v$

        \If{$u\in V_i$}{

            $N_i(v) \gets N_i(v) \cup \{u\}$

            \textbf{break}

        }
    
        $m\gets \min( \frac{6}{5}m, \sqrt d)$
    }
}
\Return The set of neighbors $N_i(v)$
\end{algorithm}


\begin{proof}[Proof of Theorem \ref{thm: quantum neighborhood}]
    Set $t = \frac{\epsilon^2\Delta}{6\log n}$ and choose an equitable partition $V = V_1 \cup \cdots \cup V_t$ at random from the set of all such partitions and prepare the unitary operator $W$ on $V\otimes [n]\otimes (V\cup \{0\})$ that does the following:
    \begin{equation}\label{eqn: partition operator}
    W|v, j, u\rangle = \begin{cases}
    -|v, j, u\rangle, &\text{ if }u,v\in V_i\text{ for some }i\\
    |v, j, u\rangle,&\text{ otherwise}
    \end{cases}.
    \end{equation}
    In other words, $W$ indicates whether or not $u$ and $v$ lie in the same part.
    Constructing $W$ requires no knowledge of $E(G)$, and hence, uses no neighborhood or adjacency queries.

    We proceed part by part, coloring each vertex in that part one by one.
    To find the neighbors of $v\in V_i$ that also live in $V_i$, start by querying the degree $d(v)$ and then prepare the uniform superposition over the indices of its neighbors:
    \[
    |U_v\rangle = d(v)^{-1/2}\sum_{j=1}^{d(v)}|v, j, 0\rangle.
    \]

    Write $N(v) = \{u_1, \ldots, u_{d(v)}\}$ and suppose $v$ has $r(v)$ neighbors in its part $V_i$.
    If we define the states $|G\rangle$ and $|B\rangle$ by
    \[
    |G\rangle  = r(v)^{-1/2}\sum_{j: u_j \in V_i} |v, j, 0\rangle,\qquad |B\rangle = \big(d(v)-r(v)\big)^{-1/2}\sum_{j: u_j\notin V_i}|v, j, 0\rangle,
    \]
    then Grover's algorithm finds one of these $r(v)$ neighbors by manipulating the state $|U_v\rangle$ in the $|G\rangle,|B\rangle$-plane until it lies on (or sufficiently close to) $|G\rangle$.
    We accomplish this by repeatedly reflecting the current state about $|B\rangle$ and then about $|U_v\rangle$.
    Reflection through $|U_v\rangle$ is implemented by the unitary map
    \[
    R_U = 2|U_v\rangle \langle U_v| - I,
    \]
    and reflection through $|B\rangle$ is implemented with the quantum neighborhood oracle  $\mathcal O_N$ (\ref{eqn: quantum neighborhood query}) by
    \[
    R_B = \mathcal O_N W \mathcal O_N.
    \]
    The inner \texttt{while} loop of \textbf{GroverNeighbors} is the implementation of Grover's algorithm from \cite{boyer1998tight}, which outputs one of the $r(v)$ desired neighbors using $O(\sqrt{d(v)}) = O(\sqrt \Delta)$ queries, succeeding with probability at least $\frac 14$.

    
    Given that this inner loop returns a neighbor of $v$ that lives in $V_i$, it is uniformly distributed over these neighbors.
    Consequently, if we rerun Grover's algorithm until we have found all $r(v)$ of $v$'s neighbors that live in $V_i$, we obtain the modified coupon collector process from the setting of Lemma \ref{lem: coupon collector}.
    Applying this lemma with $p = \frac{1}{4}$, $k = (1+\epsilon)\frac \Delta t$ and $C = \frac{2\log n}{\log \log n}$, we see that $8k \log k \cdot \frac{ \log n}{\log \log n}$ Grover searches fail to find all neighbors in $V_i$ with probability less than $\frac{1}{4n}$.
    Once \textbf{GroverNeighbors} finds all relevant neighbors, \textbf{QuantumNeighborhoodColor} simply assigns $v$ a color not seen among these neighbors.

    This procedure fails to produce a proper $(1+\epsilon)\Delta$ coloring of $G$ only if, for some $i$ and $v\in V_i$, we fail to find all neighbors of $v$ that lie in $V_i$.
    When we include the possibility that the maximum degree of some $G[V_i]$ exceeds $(1+\epsilon)\frac \Delta t$ (Lemma \ref{lem: chernoff_degree}), the total failure probability is at most $\frac 14  + \frac 1n \leq \frac 13$.
    Upon totalling up the maximum number of iterations of the loops within \textbf{GroverNeighbors}, we see that the number of calls to the quantum neighborhood oracle $\mathcal O_N$ is at most
    \begin{align*}
        n\cdot O(\sqrt \Delta) \cdot O\left( \frac \Delta t \log \frac \Delta t \cdot \frac{ \log n}{\log \log n} \right) = O\left(\frac{1}{\epsilon^2} n (\log n)^2 \sqrt \Delta  \right).
    \end{align*}

\end{proof}

\section{Discussion}
We would like to mention a few possibilities for future work.
\begin{itemize}
    \item \textbf{Improvement in $(\Delta+1)$-Coloring:} If we could obtain a $\tilde{O}(n\sqrt{\Delta})$ quantum neighborhood query algorithm for $(\Delta+1)$-coloring, we could combine it with our $\tilde{O}(n^{3/2}/\sqrt{\Delta})$ quantum adjacency query algorithm to achieve a $\tilde{O}(n^{5/4})$ query algorithm for $(\Delta+1)$-coloring.
    The factor of $(1+\epsilon)$ in our algorithm arises from the step where we randomly partition $V(G)$ to ensure that no vertex has too many neighbors within its own partition. However, using a random partition is unlikely to ensure that each vertex has at most $\frac{\Delta+1}{t}$ neighbors within its part. For the partitioning approach to be effective in $(\Delta+1)$-coloring, we suspect a more sophisticated partitioning method would be required.
    \item \textbf{Lower Bound for Quantum Algorithm:} In the classical setting, it is known that $\Omega(n^{3/2})$ queries are required for $(\Delta+1)$-coloring, and thus our classical algorithm is tight up to a $\sqrt{\log n}$ factor. However, in the quantum setting, a polynomial gap remains.  Morris and Song \cite{morrisSongColoring} showed that $\Omega(n)$ quantum adjacency queries are necessary for $O(\Delta)$-coloring by considering the optimality of Grover’s search algorithm, and we would like to know whether this bound can be achieved.
    \item \textbf{Other Computational Models:} It is worth noting that the query model is not the only interesting model in this area. The authors of \cite{assadiChenKhannaSublinear} and \cite{AlonAssadi} have also studied $(\Delta+1)$-coloring in the streaming and massively parallel computation (MPC) models, and the authors of \cite{ChangCongestedClique,parterCongest} consider the congested clique and local computational models.
    Could our algorithm, with appropriate modifications, be adapted to work in these models?
\end{itemize}

\bibliographystyle{abbrv}
\bibliography{coloring}

\end{document}